\documentclass[a4paper,UKenglish]{lipics-v2016}
\usepackage{microtype}
\usepackage{graphicx}  
\usepackage{mleftright}
\usepackage{comment}
\usepackage{xspace}
\usepackage{nameref}
\usepackage{cleveref}
\usepackage{color}

\usepackage{ifthen}
\usepackage[boxed]{algorithm}
\usepackage[noend]{algorithmic}
\usepackage{wrapfig}
\usepackage[algo2e,linesnumbered]{algorithm2e}
\usepackage{paralist}
\usepackage{threeparttable}

\usepackage{tikz}
\usetikzlibrary{backgrounds,positioning,fit,patterns}

\newboolean{short}
\setboolean{short}{false}

\newcommand{\shortOnly}[1]{\ifthenelse{\boolean{short}}{#1}{}}
\newcommand{\onlyShort}[1]{\ifthenelse{\boolean{short}}{#1}{}}
\newcommand{\longOnly}[1]{\ifthenelse{\boolean{short}}{}{#1}}
\newcommand{\onlyLong}[1]{\ifthenelse{\boolean{short}}{}{#1}}
\newcommand{\shortLong}[2]{\ifthenelse{\boolean{short}}{#2}{#1}}
\newcommand{\longShort}[2]{\ifthenelse{\boolean{short}}{#2}{#1}} 

\usepackage{amsmath}   
\usepackage{amsfonts}   
\usepackage{amssymb}   
\usepackage{amsthm}     
\usepackage{amsbsy}     
\usepackage{mathrsfs}

\def\polylog{\operatorname{polylog}}

\newcommand{\undecided}{\textsc{undecided}}
\newcommand{\catone}{\textsc{category-1}}
\newcommand{\cattwo}{\textsc{category-2}}
\newcommand{\catthree}{\textsc{category-3}}

\newtheorem{observation}{Observation}

\def\polylog{\operatorname{polylog}}
\def\cA{\mathcal{A}}
\def\cB{\mathcal{B}}

\let\originalleft\left
\let\originalright\right
\renewcommand{\left}{\mathopen{}\mathclose\bgroup\originalleft}
\renewcommand{\right}{\aftergroup\egroup\originalright}

\newcommand{\Prob}[1]{\text{Pr}\left[#1\right]\xspace}

\newcommand{\var}[1]{\text{Var}\left[#1\right]\xspace}

\newcommand{\poly}{\operatorname{poly}}

\renewcommand{\paragraph}[1]{\medskip\noindent{\bf #1.}\xspace}

\renewcommand{\le}{\leqslant}
\renewcommand{\ge}{\geqslant}
\renewcommand{\geq}{\geqslant}
\renewcommand{\leq}{\leqslant}

\title{Symmetry Breaking in the \textsc{Congest} Model: Time- and Message-Efficient Algorithms for Ruling Sets}
\titlerunning{Symmetry Breaking in the \textsc{Congest} Model} 

\author[1]{Shreyas Pai}
\author[2]{Gopal Pandurangan}
\author[1]{Sriram V.~Pemmaraju}
\author[1]{Talal Riaz}
\author[3]{Peter Robinson}
\affil[1]{Department of Computer Science, The University of Iowa, Iowa City, IA 52242, USA. \hbox{E-mail}:~{\tt \{shreyas-pai, sriram-pemmaraju, talal-riaz\}@uiowa.edu}. Suported in part by NSF grant CCF-1318166.}
\affil[2]{Department of Computer Science, University of Houston, Houston, TX 77204, USA. \hbox{E-mail}:~{\tt gopalpandurangan@gmail.com}. Supported, in part, by NSF grants CCF-1527867, CCF-1540512, and IIS-1633720.}
\affil[3]{Department of Computer Science, Royal Holloway, University of London, UK.  \hbox{E-mail}:~{\tt  peter.robinson@rhul.ac.uk}.}
\authorrunning{S. Pai, G. Pandurangan, S.\,V. Pemmaraju, T. Riaz, and P. Robinson} 

\Copyright{Shreyas Pai, Gopal Pandurangan, Sriram V.~Pemmaraju, Talal Riaz, Peter Robinson}


\begin{document}
\maketitle

\begin{abstract}

We study local symmetry breaking problems in the \textsc{Congest} model, 
focusing on ruling set problems, which generalize the fundamental Maximal Independent 
Set (MIS) problem. The \textit{time (round) complexity} of MIS (and ruling sets) have attracted much attention
in the \textsc{Local} model. Indeed, recent results (Barenboim et al., FOCS 2012, Ghaffari SODA 2016) for the MIS problem have tried to break the long-standing
$O(\log n)$-round ``barrier'' achieved by Luby's algorithm, but these yield $o(\log n)$-round complexity only when the maximum degree $\Delta$ 
is somewhat small relative to $n$. 
\textit{More importantly, these results apply only in the \textsc{Local} model.}
In fact, the best known time bound in the \textsc{Congest} model is still $O(\log n)$ (via Luby's algorithm) even for somewhat small $\Delta$.
Furthermore, \textit{message complexity} has been largely ignored in the context of local
symmetry breaking. Luby's algorithm takes $O(m)$ messages on $m$-edge graphs
and this is the best known bound with respect to messages. Our work is motivated 
by the following central question: can we break the $\Theta(\log n)$ time complexity barrier 
and the $\Theta(m)$ message complexity barrier in the
\textsc{Congest} model for MIS or closely-related symmetry breaking problems?

This paper presents progress towards this question for the distributed ruling set problem in the \textsc{Congest} model. 
A \textit{$\beta$-ruling set} is an independent set such that every node in the graph is at most $\beta$ hops from a node in the 
independent set. We present the following results:
\begin{itemize}
\item {\em Time Complexity:} We show that we can break the $O(\log n)$ ``barrier'' for 2- and 3-ruling sets.
We compute 3-ruling sets in $O\left(\frac{\log n}{\log \log n}\right)$ rounds with high probability (whp). 
More generally we show that 2-ruling sets can be computed in $O\left(\log \Delta \cdot (\log n)^{1/2 + \varepsilon} + \frac{\log n}{\log\log n}\right)$ rounds for any $\varepsilon > 0$,
which is $o(\log n)$ for a wide range of $\Delta$ values (e.g., $\Delta = 2^{(\log n)^{1/2-\varepsilon}}$).
These are the first 2- and 3-ruling set algorithms to improve over the $O(\log n)$-round complexity of Luby's algorithm in the \textsc{Congest} model.
\item {\em Message Complexity:}  We show an $\Omega(n^2)$ lower bound on the message complexity of computing an MIS (i.e., 1-ruling set) which holds also for randomized algorithms 
and present a contrast to this by showing a randomized algorithm for 2-ruling sets that, whp, uses  
only $O(n \log^2 n)$ messages and runs in $O(\Delta \log n)$ rounds. This
is the first message-efficient algorithm known for ruling sets, which has message complexity nearly 
linear in $n$ (which is optimal up to a polylogarithmic factor).
\end{itemize} 
Our results are a step toward understanding the  time and message complexity of symmetry breaking problems in the \textsc{Congest} model.
\end{abstract}


\section{Introduction}
\label{sec:intro}
The \textit{maximal independent set (MIS)} problem is one of the fundamental problems in distributed
computing because it is a simple and elegant abstraction of ``local symmetry breaking,'' an issue
that arises repeatedly in many distributed computing problems.
About 30 years ago Alon, Babai, and Itai \cite{AlonBabaiItai} and Luby \cite{LubySICOMP86} 
presented a randomized algorithm for MIS, running on $n$-node graphs in $O(\log n)$ rounds with high probability (whp)\footnote{Throughout, we 
use ``with high probability (whp)'' to mean with probability at least $1 - 1/n^c$, for some $c \ge 1$.}.
Since then the MIS problem has been studied extensively and
recently, there has been some exciting progress in designing faster MIS algorithms.
For $n$-node graphs with maximum degree $\Delta$, Ghaffari \cite{GhaffariSODA16} presented 
an MIS algorithm running in $O(\log \Delta) + 2^{O(\sqrt{\log\log n})}$ rounds, improving over the algorithm of Barenboim et al.~\cite{BarenboimEPSJACM2016}
that runs in $O(\log^2 \Delta) + 2^{O(\sqrt{\log\log n})}$ rounds.
Ghaffari's MIS algorithm is the first MIS algorithm to improve over the round complexity
of Luby's algorithm when $\Delta = 2^{o(\log n)}$ and $\Delta$ is bounded below by $\Omega(\log n)$.\footnote{For $\Delta = o(\log n)$, the deterministic MIS algorithm of Barenboim, Elkin, and 
Kuhn \cite{BarenboimEKSICOMP2014} that runs $O(\Delta + \log^* n)$ rounds is faster than Luby's algorithm.}

While the results of Ghaffari and Barenboim et al.~constitute a significant improvement in our understanding of the round complexity of the MIS problem, it should
be noted that both of these results are in the \textsc{Local} model.
The \textsc{Local} model \cite{PelegBook} is a synchronous, message-passing model of distributed computing in which 
\textit{messages can be arbitrarily large}.
Luby's algorithm, on the other hand, is in the \textsc{Congest} model \cite{PelegBook} and uses small messages, i.e., messages 
that are $O(\log n)$ bits or $O(1)$ words in size. 
In fact, to date, {\em Luby's algorithm is the fastest known MIS algorithm in the \textsc{Congest} model}; this is the case even when $\Delta$ is 
between $\Omega(\log n)$ and $2^{o(\log n)}$.
For example, for the class of graphs with $\Delta = 2^{O(\sqrt{\log n})}$, Ghaffari's MIS algorithm runs
in $O(\sqrt{\log n})$ rounds whp in the \textsc{Local} model, but we don't know how to compute
an MIS for this class of graphs in $o(\log n)$ rounds in the \textsc{Congest} model.
It should be further noted that the MIS algorithms of Ghaffari and Barenboim et al.~use messages of size
$O(\text{poly}(\Delta) \log n)$ (see Theorem 3.5 in \cite{BarenboimEPSJACM2016}), which can be much larger than
the $O(\log n)$-sized messages allowed in the \textsc{Congest} model; in fact
these algorithms do not run within the claimed number of rounds even if
messages of size $O(\poly(\log n))$ were allowed.
Furthermore, large messages arise in these algorithms from a topology-gathering step in which cluster-leaders
gather the entire topology of their clusters in order to compute an MIS of their cluster -- this 
step seems fundamental to these algorithms and there does not seem to be an efficient way to simulate this
step in the \textsc{Congest} model.

Ruling sets are a natural generalization of MIS and have also been well-studied in the \textsc{Local} model. 
An \textit{$(\alpha, \beta)$-ruling set} \cite{GPS87} is a node-subset $T$ such that (i) any two distinct nodes in $T$ are at
least $\alpha$ hops apart in $G$ and (ii) every node in the graph is at most $\beta$ hops from some node in $T$.
A $(2, \beta)$-ruling set is an independent set and since such ruling sets are the main focus of this paper, 
we use the shorthand \textit{$\beta$-ruling sets} to refer to $(2, \beta)$-ruling sets. (Using this terminology 
an MIS is just a 1-ruling set.)
The above mentioned MIS results due to Barenboim et al.~and Ghaffari have also led to the 
\textit{sublogarithmic}-round algorithms for $\beta$-ruling sets for $\beta \ge 2$. 
The earliest instance of such a result was the algorithm of Kothapalli and Pemmaraju \cite{KP12} that 
computed a 2-ruling set in $O(\sqrt{\log \Delta} \cdot (\log n)^{1/4})$ rounds by 
using an earlier version of the Barenboim et al.~\cite{BEPS12FOCS} MIS algorithm.
There have been several further improvements in the running time of ruling set algorithms culminating in the 
$O(\beta \log^{1/\beta}\Delta) + 2^{O(\sqrt{\log\log n})}$ round $\beta$-ruling set algorithm of
Ghaffari \cite{GhaffariSODA16}.
This result is based on a recursive sparsification procedure of Bisht et al.~\cite{BishtKPPODC2014}
that reduces the $\beta$-ruling set problem on graphs with maximum degree $\Delta$ to an MIS problem
on graphs with degree much smaller.
Ghaffari's $\beta$-ruling set result is also interesting because it identifies a separation between 2-ruling sets and 
MIS (1-ruling sets).  
This follows from the lower bound of $\Omega\left(\min\left\{\sqrt{\frac{\log n}{\log\log n}}, \frac{\log \Delta}{\log\log \Delta}\right\}\right)$ for MIS due to Kuhn et al.~\cite{KuhnMoscibrodaWattenhoferFull}.
Again, we emphasize here that all of these improvements for ruling set algorithms are \textit{only in the \textsc{Local} model} because these ruling set algorithms rely on \textsc{Local}-model MIS algorithms to ``finish off'' the processing
of small degree subgraphs.
As far as we know, prior to the current work there has been no $o(\log n)$-round, $\beta$-ruling set algorithm in the 
\textsc{Congest} model for any $\beta = O(1)$.

The focus of all the above results has been on the time (round) complexity. \textit{Message complexity}, on the other hand, has been largely ignored in the context of local
symmetry breaking problems such as MIS and ruling sets. 
For a graph with $m$ edges, Luby's algorithm uses $O(m)$ messages in the \textsc{Congest} model and until now there has been 
no MIS or ruling set algorithm that uses $o(m)$ messages. We note that the ruling set algorithm of Goldberg et al.~\cite{GPS87} 
which can be implemented in the \textsc{Congest} model \cite{danupon} also takes at least $\Omega(m)$ messages.

The focus of this paper is symmetry breaking problems in the \textsc{Congest} model and the specific question that
motivates our work is whether we can go beyond Luby's algorithm in the \textsc{Congest} model for MIS or any closely-related symmetry breaking
problems such as \textit{ruling sets}. In particular, {\em can we break the $\Theta(\log n)$ 
time complexity barrier and the $\Theta(m)$ message complexity barrier, in the \textsc{Congest} model for MIS and ruling sets?}
In many applications, especially in resource-constrained communication networks and in distributed processing of large-scale data
it is important to design distributed algorithms that have  low time complexity as well as message complexity. In particular, optimizing messages as well as time has direct applications
to the performance of distributed algorithms  in other models such as the $k$-machine model \cite{soda15}.

We present two sets of results, one set focusing on time (round) complexity and the other on message complexity.

\noindent
 {\bf 1.}  {\em Time complexity:} (cf. Section \ref{sec:time}) We first show that 2-ruling sets can be computed in the \textsc{Congest} model in 
$O\left(\log \Delta \cdot (\log n)^{1/2 + \varepsilon} + \frac{\log n}{\log\log n}\right)$ rounds whp for $n$-node graphs with maximum degree
$\Delta$ and for any $\varepsilon > 0$.
This is the first algorithm to improve over Luby's algorithm, by running in $o(\log n)$ rounds in the \textsc{Congest} model, 
for a wide range of values of $\Delta$.
Specifically our algorithm runs in $o(\log n)$ rounds for $\Delta$ bounded above by $2^{(\log n)^{1/2-\varepsilon}}$ 
for any value of $\varepsilon > 0$.
Using this 2-ruling set algorithm as a subroutine, we show how to compute 3-ruling sets (for any graph) in $O\left(\frac{\log  n}{\log\log n}\right)$ rounds whp
in the \textsc{Congest} model. We also present a simple 5-ruling set algorithm based on Ghaffari's MIS
algorithm that runs in $O(\sqrt{\log n})$ rounds in the \textsc{Congest} model.

\noindent
{\bf 2.}  {\em Message complexity:} (cf. Sections \ref{sec:msg} and \ref{sec:lower}) We show that $\Omega(n^2)$ is a fundamental lower bound for computing an MIS (i.e., 1-ruling set)
by showing that there exists graphs (with $m = \Theta(n^2)$ edges)  where any distributed MIS algorithm needs $\Omega(n^2)$ messages. 
In contrast, we show that 2-ruling sets can be computed using significantly smaller message complexity. In particular,
we present a randomized 2-ruling set algorithm that, whp, uses $O(n \log^2 n)$ messages and runs in $O(\Delta \log n)$ rounds. This
is the first $o(m)$-message algorithm known for ruling sets, which takes near-linear (in $n$) message complexity.
This message bound is tight up to a polylogarithmic factor, since we show that any $O(1)$-ruling set (randomized) algorithm that succeeds with probability $1-o(1)$ requires $\Omega(n)$ messages in the worst case. We also present a simple 2-ruling set algorithm that uses $O(n^{1.5}\log n)$ messages, but  runs faster --- in $O(\log n)$ rounds.

Our results make progress towards understanding the complexity of symmetry breaking, in particular with respect to ruling sets, in the \textsc{Congest} model.
With regards to time complexity, our results, for the first time, show that one can obtain $o(\log n)$ round algorithms for ruling sets  in the \textsc{Congest} model.
With regards to message complexity, our results are (essentially) tight: while MIS needs quadratic (in $n$) messages in the worst case, 2-ruling sets can
be computed using near-linear (in $n$) messages. We discuss key problems left open by our work in Section \ref{sec:conc}.
\onlyShort{Other related work and omitted proofs  can be found in the full version (in the Appendix).}

\subsection{Distributed Computing Model} \label{sec:model}
We consider the standard synchronous \textsc{Congest} model \cite{PelegBook}  described as follows.

We are given a distributed network of $n$ nodes, modeled as an undirected graph $G$.
Each node hosts a processor with limited initial knowledge. 
We assume that nodes have unique \texttt{ID}s (this is not essential, but
simplifies presentation), and at the beginning of the computation each 
node is provided its \texttt{ID} as input.
Thus, a node has only {\em local} knowledge\footnote{Our near-linear message-efficient algorithm (Section \ref{sec:msg}) does not require knowledge
of $n$ or $\Delta$, whereas our time-efficient algorithms (Section \ref{sec:time}) assume knowledge of $n$ and $\Delta$
(otherwise it will work up to a given $\Delta$).}. Specifically we assume that each node has ports (each port having a unique port number); each incident edge is connected
to one distinct port. 
This model is  referred to as the {\em clean network model} in \cite{PelegBook} and is also sometimes referred to as the
$KT_0$ model, i.e., the initial (K)nowledge of
all nodes is restricted (T)ill radius 0 (i.e., just the local knowledge) \cite{awerbuch}.  

Nodes are allowed to communicate through the edges of the graph $G$ and it is assumed
that communication is synchronous and occurs in discrete rounds (time steps). 
In each round, each node can perform some local computation including accessing a private source of randomness, and can exchange (possibly distinct) $O(\log n)$-bit messages with each of its neighboring nodes. 
This model of distributed computation is called the $\textsc{Congest}(\log n)$
model or simply the \textsc{Congest} model \cite{PelegBook}. 

\onlyLong{
\subsection{Related Work}
\label{section:relatedWork}

As one would expect, \textsc{Congest} model symmetry breaking algorithms are easier for \textit{sparse graphs}.
There is a deterministic $(\Delta + 1)$-coloring algorithm due to Barenboim, Elkin, and Kuhn \cite{BarenboimEKSICOMP2014} that runs in the \textsc{Congest} model in $O(\Delta) + \frac{1}{2} \log^* n$ rounds.
This can be used to obtain an $o(\log n)$-round \textsc{Congest} model MIS algorithm, when
$\Delta = o(\log n)$. 
For trees, Lenzen and Wattenhofer \cite{LenzenWattenhofer} presented an MIS algorithm that runs in
$O(\sqrt{\log n} \log\log n)$ rounds whp in the \textsc{Congest} model.
More generally, for graphs with arboricity bounded above by $\alpha$, Pemmaraju and Riaz \cite{PemmarajuRiazOPODIS2016} present
an MIS algorithm that runs in $O(\mbox{poly}(\alpha) \cdot \sqrt{\log n \log\log n})$ rounds in the
\textsc{Congest} model.
Other research that is relevant to ruling sets, but is only in the \textsc{Local} model,
includes the multi-trials technique of Schneider and Wattenhofer \cite{wattenhofer2010podc} and 
the deterministic ruling set algorithms of Schneider et al.~\cite{SchneiderEW13}.

As mentioned earlier, in the context of local symmetry breaking problems such as MIS or ruling sets, message complexity 
has not received much attention.
However, in the context of global problems (i.e., problems where one needs to traverse the entire network and, hence, take at least $\Omega(D)$ time) such as leader election (which can be thought as a ``global" symmetry breaking) and minimum spanning tree (MST),
message complexity has been very well studied. Kutten et al.~\cite{DBLP:journals/jacm/KuttenPP0T15} showed that $\Omega(m)$ is a message lower bound for leader election  and this applies
to randomized Monte-Carlo algorithms as well. This lower bound also applies to the 
Broadcast and MST problems. 
In a similar spirit, in this paper, we show that in general, $\Omega(m)$ is a message lower bound for the MIS problem as well.
In contrast, we show that this lower bound does not hold for ruling sets which admit a near-linear (in $n$) message complexity.

It is important to point out that the current paper as well as most prior work on 
leader election and MST 
~\cite{awerbuch-optimal,chin-almostlinear,DistMst:Gallager,DistMst:Garay,kutten-domset,gafni-election,elkin-faster,stoc17,DBLP:journals/jacm/KuttenPP0T15})
assume the $KT_0$ model. 
\onlyLong{However, one can also consider a stronger model where nodes have initial knowledge of the identity of their neighbors.
This model is called the {\em $KT_1$ model}. 
Awerbuch et al.~\cite{awerbuch} show that $\Omega(m)$ is a message lower bound for MST  for the $KT_1$ model,  
if one allows only comparison-based algorithms (i.e., algorithms that can operate on IDs only by comparing them); this lower bound for comparison-based algorithms
applies to \textit{randomized} algorithms as well. 
Awerbuch et al.~\cite{awerbuch} also show that the $\Omega(m)$ message lower bound  applies even to non-comparison based (in particular, algorithms that can perform arbitrary local computations)  {\em deterministic} algorithms 
in the \textsc{Congest} model  that terminate in a time bound that depends only on the graph topology (e.g., a function of $n$). 
On the other hand, for {\em randomized non-comparison-based} algorithms, it turns out that the message lower bound
of $\Omega(m)$ does not  apply in the $KT_1$ model. 
King et al.~\cite{KingKT15} showed a surprising and elegant result (also see \cite{icdcn17}): in the $KT_1$ model one can give a randomized Monte Carlo algorithm to construct a MST or a spanning tree in $\tilde{O}(n)$ messages  ($\Omega(n)$ is a message lower bound) and in $\tilde{O}(n)$ time (this algorithm uses randomness and is not comparison-based). While this algorithm shows that one can get $o(m)$ message complexity (when $m = \omega(n \polylog n)$),  it  is {\em not} time-optimal (it can take up to $\tilde{O}(n)$ rounds).}

One can also use the King et al.~algorithm to build a spanning tree using $\tilde{O}(n)$ messages and then use
time encoding (see e.g., \cite{podc15,sirocco16}) to collect the entire graph topology at the root of the 
spanning tree. 
Hence, using this approach any problem (including, MST, MIS, ruling sets, etc.) can be solved 
using $\tilde{O}(n)$ 
messages in the $KT_1$ model. However, this is highly time inefficient as it takes exponential (in $n$) rounds.
}

\subsection{Technical Overview}

\subsubsection{Time Bounds}
The MIS algorithms of Barenboim et al.~\cite{BarenboimEPSJACM2016} and Ghaffari \cite{GhaffariSODA16} use a 
2-phase strategy, attributed to Beck \cite{Beck1991}, who used it in his algorithmic version of 
the Lov\'{a}sz Local Lemma. In the first phase, some number of iterations of a Luby-type ``base algorithm'' are run 
(in the \textsc{Congest} model). During this phase, some nodes join the MIS and these nodes and their 
neighbors become inactive. The first phase is run until the graph is ``shattered'', i.e., the nodes that
remain active induce a number of ``small'' connected components.
Once the graph is ``shattered'', the algorithm switches to the second, deterministic phase to ``finish 
off'' the problem in the remaining small components. 
It is this second phase that relies critically on the use of the \textsc{Local} model in order to run fast.

In general, in the \textsc{Congest} model it is not clear how to take advantage of low degree or low 
diameter or small size of a connected component to solve symmetry-breaking
problems (MIS or ruling sets) faster than the $O(\log n)$-round bound provided by Luby's algorithm.
In both Barenboim et al.~\cite{BarenboimEPSJACM2016} and Ghaffari \cite{GhaffariSODA16}, a key ingredient of the
second ``finish-off'' phase is the deterministic network decomposition algorithm of
Panconesi and Srinivasan \cite{PS96} that can be used to compute an MIS in $O(2^{\sqrt{\log s}})$ rounds 
on a graph with $s$ nodes in the \textsc{Local} model.
If one can get connected components of size $O(\mbox{poly}(\log n))$ then it is possible to finish the 
rest of the algorithm in $2^{O(\sqrt{\log\log n})}$ rounds and this is indeed the source of the 
``$2^{O(\sqrt{\log\log n})}$'' term in the round complexity of these MIS algorithms.
In fact, the Panconesi-Srinivasan network decomposition algorithm itself runs in the \textsc{Congest} model,
but once the network has been decomposed into small diameter clusters then algorithms simply resort to
gathering the entire topology of a cluster at a cluster-leader and this requires large messages.
Currently, there seem to be no techniques for symmetry breaking problems in the \textsc{Congest} model that are able to take 
advantage of the diameter of a network being small.
As far as we know, there is no $o(\log n)$-round $O(1)$-ruling set algorithm in the \textsc{Congest} model \textit{even for constant-diameter graphs}, for any constant larger than 1.
To obtain our sublogarithmic $\beta$-ruling set algorithms (for $\beta = 2, 3, 5$), we use simple 
\textit{greedy} MIS and 2-ruling set
algorithms to process ``small'' subgraphs in the final stages of algorithm.
These greedy algorithms just exchange $O(\log n)$-bit \texttt{ID}s with neighbors and run in the \textsc{Congest}
model, but they can take $\Theta(s)$ rounds in the worst case, where $s$ is the length of the \textit{longest} path in
the subgraph. So our main technical contribution is to show that it is possible to do a randomized shattering of the graph so 
that none of the fragments have any long paths.

\subsubsection{Message Bounds}


As mentioned earlier, our message complexity lower bound for MIS and the contrasting the upper bound for 2-ruling set
show a clear separation between these two problems.
At a high-level, our lower bound argument
exploits the idea of ``bridge crossing" (similar to \cite{DBLP:journals/jacm/KuttenPP0T15}) whose intuition is as follows.
We consider two types of related graphs: (1) a {\em complete bipartite graph} and
(2) a  {\em random bridge graph} which consists of a two  (almost-)complete bipartite graphs connected by  two ``bridge" edges chosen randomly  (see Figure \ref{fig:messagelb} and  Section
\ref{sec:lower} for a detailed description of the construction). Note that the MIS in a complete bipartite graph is exactly the set of all nodes
belonging to one part of the partition. The crucial observation is that if no messages are sent over bridge edges, then the bipartite graphs
on either side of the bridge edges behave identically which can result in choosing adjacent nodes in MIS, a violation.
In particular,  we show that if an algorithm sends $o(n^2)$ messages, then with probability at least $1-o(1)$
 that there will be {\em no} message sent over the bridge edges and by symmetry, with probability at least $1/2$, 
 two nodes that are connected by the bridge edge will be chosen to be in the MIS.

Our 2-ruling set algorithm with low-message-complexity crucially uses the fact that, unlike in an MIS, in a 2-ruling set there are 3 
categories of nodes: $\catone$ (nodes that are in the independent set), 
$\cattwo$ (nodes that are neighbors of $\catone$) and $\catthree$ nodes (nodes that
 are neighbors of $\cattwo$, but not neighbors of $\catone$). 
Our algorithm, inspired by Luby's MIS algorithm, uses three main ideas. First, $\cattwo$ and $\catthree$ nodes don't initiate messages; only undecided nodes (i.e., nodes whose category are not yet decided) initiate messages. Second, an undecided node   does ``checking sampling" (cf. Algorithm \ref{alg:2-ruling-set}) first before it does local broadcast, i.e., it samples a few of its neighbors to see if they are any $\cattwo$ nodes; if so
it becomes a $\catthree$ node immediately. 
Third, an undecided node tries to enter the ruling set with probability that is \textit{always} inversely proportional to its
{\em original} degree, i.e., $\Theta(1/d(v))$, where $d(v)$ is the degree of $v$.
This is unlike in Luby's algorithm, where the
marking probability is inversely proportional to its \textit{current  degree}.
These ideas along with an \textit{amortized} charging argument \cite{Algorithm:Cormen} yield our result: 
an algorithm using $O(n \log^2 n)$ messages and running in $O(\Delta \log n)$ rounds.

\section{Time-Efficient Ruling Set Algorithms in the \textsc{Congest} model}
\label{sec:time}
The main result of this section is a 2-ruling set algorithm in the \textsc{Congest} model 
that runs in $O\left(\log \Delta \cdot (\log n)^{1/2 + \varepsilon} + \frac{\log n}{\log\log n}\right)$ rounds whp, for any constant $\varepsilon > 0$, 
on $n$-node graphs with maximum degree $\Delta$.
An implication of this result is that for graphs with $\Delta = 2^{O((\log n)^{1/2-\varepsilon})}$ for any $\varepsilon > 0$, we can compute a 2-ruling set in 
$O\left(\frac{\log n}{\log\log n}\right)$ rounds in the \textsc{Congest} model.
A second implication is that using this 2-ruling set algorithm as a subroutine, we can compute a 
3-ruling set for \textit{any} graph in $O\left(\frac{\log n}{\log\log n}\right)$ rounds whp in the 
\textsc{Congest} model.
These are the first sublogarithmic-round \textsc{Congest} model algorithms for 2-ruling sets (for a wide range
of $\Delta$) and 3-ruling sets.
Combining some of the techniques used in these algorithms with the first phase of
Ghaffari's MIS algorithm \cite{GhaffariSODA16}, we also show that a 5-ruling set can 
be computed in $O(\sqrt{\log n})$ rounds whp in the \textsc{Congest} model.
 
\subsection{The 2-ruling Set Algorithm}
\label{sec:2-ruling-time}
Our 2-ruling set algorithm (described in pseudocode below) takes as input an $n$-node graph with
maximum degree $\Delta \le 2^{\sqrt{\log n}}$, along with a parameter $\varepsilon > 0$.
For $\Delta > 2^{\sqrt{\log n}}$, we simply execute Luby's MIS algorithm to solve the problem.
The algorithm consists of $\lceil \log \Delta \rceil$
\textit{scales} and in scale $t$, $1 \le t \le \lceil \log \Delta \rceil$, nodes with degrees
at most $\Delta_t := \Delta/2^{t-1}$ are processed.
Each scale consists of $\Theta(\log^{1/2 + \varepsilon} n)$ \textit{iterations}.
In an iteration $i$, in scale $t$, each undecided node independently joins a set $M_{i,t}$ with
probability $1/(\Delta_t \cdot \log^{\varepsilon} n)$ (Line \ref{alg1:marking}).
Neighbors of nodes in $M_{i,t}$, that are themselves not in $M_{i,t}$, are set aside and placed
in a set $W_{i,t}$ (Lines \ref{alg1:ifwit}-\ref{alg1:endifwit}). The nodes in $M_{i,t} \cup W_{i,t}$ have decided their fate and we continue to
process the undecided nodes.
At the end of all the iterations in a scale $t$, any undecided node that still has $\Delta_t/2$ or more 
undecided neighbors is placed in a ``bad'' set $B_t$ for that scale (Line \ref{alg1:badset}), thus effectively deciding the fate of all nodes
with degree at least $\Delta_t/2$.
We now process the set of scale-$t$ ``bad'' nodes, $B_t$, by simply running 
a greedy 2-ruling set algorithm on $B_t$ (Line \ref{alg1:2rulset}).
We also need to process the sets $M_{i,t}$ (Line \ref{alg1:MIS}) and for that we rely on 
a greedy 1-ruling set algorithm (i.e., a greedy MIS algorithm).
Note that the $M_{i,t}$'s are all disconnected from each other since the $W_{i,t}$'s act 
as ``buffers'' around the $M_{i,t}$'s. Thus after all the scales are completed, we can compute
an MIS on all of the $M_{i,t}$'s in parallel.
Since each node in $W_{i,t}$ has a neighbor in $M_{i,t}$, this will guarantee that 
every node in $W_{i,t}$ has an independent set node at most 2 hops away.
In the following algorithm we use $deg_S(v)$ to denote the degree of a vertex $v$ in the $G[S]$, the graph
induced by $S$. 

\vfill
\RestyleAlgo{boxruled}
\begin{algorithm2e}[H]
  \caption{\textsc{2-ruling Set}(Graph $G = (V, E)$, $\varepsilon > 0$):}
  $I \leftarrow \emptyset$; $S \leftarrow V$\;
  \For {each scale $t = 1, 2,\ldots, \lceil \log \Delta \rceil$} { \label{alg1:scaleloop}
    Let $\Delta_t = \frac{\Delta}{2^{t-1}}$; $S_t \leftarrow S$\;
    \For{iteration $i = 1, 2, \ldots, \lceil c\cdot \log^{1/2 + \varepsilon} n\rceil$}{ \label{alg1:iterloop}

      Each $v\in S$ marks itself and joins  $M_{i,t}$ with probability $\frac{1}{\Delta_t \cdot \log^{\varepsilon} n}$\; \label{alg1:marking}
      \If{$v \in S$ is unmarked and a neighbor in $S$ is marked}{ \label{alg1:ifwit} $v$ joins $W_{i,t}$\; }\label{alg1:endifwit}
      $S \leftarrow S\setminus(M_{i,t} \cup W_{i,t})$\;

      }
      $B_t \leftarrow \{v \in S \mid deg_S(v) \geq \Delta_t/2\}$\; \label{alg1:badset}
      $S \leftarrow S \setminus B_t$\;
      $I \leftarrow I \cup \textsc{GreedyRulingSet}(G[S_t], B_t, 2)$\label{alg1:2rulset}\;
    }\label{alg1:endscaleloop}
    $I \leftarrow I \cup (\cup_t \cup_i \textsc{GreedyRulingSet}(G[S_t], M_{i,t}, 1))$\; \label{alg1:MIS}
    \textbf{return} $I$\;
\end{algorithm2e}

\noindent
The overall round complexity of this algorithm critically depends on the greedy 2-ruling set
algorithm terminating quickly on each $B_t$ (Line \ref{alg1:2rulset}) and the greedy 1-ruling set algorithm terminating
quickly on each $M_{i,t}$ (Line \ref{alg1:MIS}).
To be concrete, we present below a specific $\beta$-ruling set algorithm that greedily picks nodes 
by their \texttt{ID}s from a given node subset $R$.
\RestyleAlgo{boxruled}
\begin{algorithm2e}\caption{\textsc{GreedyRulingSet}(Graph $G = (V, E)$, $R \subseteq V$, integer $\beta > 0$):}
  $I \leftarrow \emptyset$; $U \leftarrow R$; \tcp{$U$ is the initial set of undecided nodes}
  \While{$U \not= \emptyset$}{
    \For{each node $v \in U$ in parallel}{
      \If{($v$ has higher \texttt{ID} than all neighbors in $U$)}{
        $I \leftarrow I \cup \{v\}$;\\
        $v$ and nodes within distance $\beta$ in $G$ are removed from $U$
      }
    }
  }
  \textbf{return} $I$
\label{algo:greedyRulingSet}
\end{algorithm2e}

\noindent
To show that the calls to this greedy ruling set algorithm terminate quickly, we introduce the notion of 
\textit{witness paths}.
If $\textsc{GreedyRulingSet}(G, R, \beta)$ runs for $p$ iterations (of the \textbf{while}-loop), then $R$
must contain a sequence of nodes $(v_1, v_2, \ldots, v_p)$ such that $v_i$, $1 \le i \le p$, joins the
independent set $I$ in iteration $i$ and node $v_i$, $1 < i \le p$, must contain an undecided node with higher \texttt{ID}
in its $1$-neighborhood in $G$, which was removed when $v_{i-1}$ and its $\beta$-neighborhood in $G$ were removed in iteration $i-1$.
We call such a sequence a \textit{witness path} for the execution of \textsc{GreedyRulingSet}.
Three simple properties of witness paths are needed in our analysis:
(i) any two nodes $v_i$ and $v_j$ in the witness path are at least $\beta +1$ hops away in $G$,
(ii) any two consecutive nodes $v_i$ and $v_{i+1}$ in the witness path are at most
$\beta +1$ hops away in $G$, and
(iii) $G[R]$ contains a simple path with $(p-1)(\beta + 1) + 1$ nodes, starting at node $v_1$, passing through
nodes $v_2, v_3, \ldots, v_{p-1}$ and ending at node $v_p$.

\noindent
To show that each $M_{i,t}$ can be processed quickly by the greedy 1-ruling set algorithm 
we show (in Lemma \ref{lemma:pathM}) that whp every witness path for the execution of the greedy
1-ruling set algorithm is short.
Similarly, to show that each $B_t$ can be processed quickly by the greedy 2-ruling set algorithm we
prove (in Lemma \ref{lemma:pathDelta}) that whp a ``bad'' set $B_t$ cannot contain a witness path
of length $\sqrt{\log n}$ or longer to the execution of the greedy 2-ruling set algorithm.
At the start of our analysis we observe that the set $S_t$, which is the set of undecided nodes at the start of scale $t$,
induces a subgraph with maximum degree $\Delta_t = \Delta/2^{t-1}$.
\begin{lemma}\label{lemma:pathM}
  For all scales $t$ and iterations $i$, $\textsc{GreedyRulingSet}(G[S_t], M_{i, t}, 1)$ runs
  in $O\left(\frac{\log n}{\varepsilon \log \log n}\right)$ rounds, whp.
\end{lemma}
\begin{proof} 
 Consider an arbitrary scale $t$ and iteration $i$.
 By Property (iii) of witness paths, there is a simple path $P$ with $(2p-1)$ nodes in $G[S_t]$, all of whose nodes have joined $M_{i, t}$.
 Due to independence of the marking step (Line \ref{alg1:marking}) the probability that all nodes in $P$ join $M_{i,t}$ is at most 
 $(1/\Delta_t \cdot \log^{\varepsilon} n)^{2p-1}$.
 Since $\Delta(G[S_t]) \le \Delta_t$, the number of simple paths with $2p-1$ nodes in $G[S_t]$ are at most $n \cdot \Delta_t^{2p-1}$. 
 Using a union bound over all candidate simple paths with $2p-1$ nodes in $G[S_t]$, we see that the probability that there exists a 
 simple path in $G[M_{i,t}]$ of length $2p-1$ is at most:
  $ n\cdot \Delta_t^{2p-1}\cdot \left (\frac{1}{\Delta_t\log^{\varepsilon} n}\right)^{2p-1}  = n\cdot\frac{1}{(\log n)^{\varepsilon (2p-1)}}.$
    Picking $p$ to be the smallest integer such that $2p-1 \ge \frac{4\log n}{\varepsilon \log\log n}$, we get
    \[\Pr(\exists \text{ a simple path with $2p-1$ nodes that joins } M_{i,t}) \leq n \cdot\frac{1}{\left(2^{\log\log n}\right)^{\varepsilon \frac{4\log n}{\varepsilon\log\log n}}} =  n\cdot\frac{1}{n^4} = \frac{1}{n^3}. \]
	We have $O(\log \Delta \cdot (\log n)^{1/2+\varepsilon})$ different $M_{i,t}$'s. Using a union bound over these $M_{i,t}$'s, we see that the probability that there exists an $M_{i,t}$ 
    containing a simple path with $2p-1$ nodes is at most $n^{-2}$.
     Thus with probability at least $1 - 1/n^2$, all of the calls to $\textsc{GreedyRulingSet}(G[S_t], M_{i,t}, 1))$ (in Line \ref{alg1:MIS})
     complete in $O\left(\frac{\log n}{\log\log n}\right)$ rounds.
\end{proof}

\begin{lemma}\label{lemma:pathDelta}
  For all scales $t$, the call to $\textsc{GreedyRulingSet}(G[S_t], B_t, 2)$ takes $O(\sqrt{\log n})$ rounds whp.
\end{lemma}
\begin{proof}
Consider a length-$p$ witness path $P$ for the execution of $\textsc{GreedyRulingSet}(G[S_t], B_t, 2)$ (Line \ref{alg1:2rulset}). 
By Property (i) of witness paths, all pairs of nodes in $P$ are at distance at least 3 from each other. 
Fix a scale $t$. We now calculate the probability that all nodes in $P$ belong to $B_t$.
Consider some node  $v \in P$. For $v$ to belong to $B_t$, it must have not marked itself in all iterations of scale $t$ 
and moreover at least $\Delta_t/2$ neighbors of $v$ in $S_t$ must not have marked themselves in  any iteration of scale $t$. 
Since the neighborhoods of any two nodes in $P$ are disjoint, the event that $v$ joins $B_t$ is independent of any other node in $P$ joining $B_t$. Therefore,
$$\Pr(P \text{ is in } B_t )  \leq \prod_{v \in P}{\Pr\left(v \text{ and at least $\Delta_t/2$ neighbors do not mark themselves in scale } t\right)}.$$
This can be bounded above by
$\prod_{v \in P}{\left(1-\frac{1}{\Delta_{t}(\log n)^{\varepsilon}}\right)^{\frac{\Delta_{t}}{2} \cdot c(\log n)^{1/2 + \varepsilon}}} \leq \exp \left(-\frac{c}{2} \cdot (\log n)^{1/2} \cdot p\right)$.
  Plugging in $p = \sqrt{\log n}$ we see that this probability is bounded above by $n^{-c/2}$.
  By Property (ii) of witness paths and the fact that $\Delta(G[S_t]) \le \Delta_t$, we know that there are at most $n \cdot (\Delta_t)^{3p}$
  length-$p$ candidate witness paths.
  Using a union bound over all of these, we get that the probability that there exists a
  a witness path that joins $B_t$ is at most $n \Delta^{3p} \cdot n^{-c/2}$.
  Plugging in $\Delta \leq 2^{\sqrt{\log n}}$ and $p = \sqrt{\log n}$ we get that this probability 
  is at most $n \cdot n^3 \cdot n^{-c/2} = n^{-c/2 + 4}.$
 Picking a large enough constant $c$ guarantees that this probability is at most $1/n^2$ and taking a final union bound over all
$\lceil \log \Delta \rceil$ scales gives us the result that all calls to $\textsc{GreedyRulingSet}(G[S_t], B_t, 2)$ take $O(\sqrt{\log n})$ rounds whp.
\end{proof}
\noindent
\onlyShort{Lemmas \ref{lemma:pathM} and \ref{lemma:pathDelta} prove upper bounds on the number of rounds it takes for the 
calls to \textsc{GreedyRulingSet} (in Lines \ref{alg1:2rulset} and \ref{alg1:MIS}). 
Now analyzing Algorithm \textsc{2-Ruling Set} is straightforward and leads to the following theorem.}
\begin{theorem}
	Algorithm \textsc{2-RulingSet} computes a 2-ruling set in the \textsc{Congest} model in 
	$O\left(\log \Delta \cdot (\log n)^{1/2 + \varepsilon} + \frac{\log n}{\varepsilon \log\log n}\right)$ rounds, whp.
\end{theorem}
\onlyLong{
\begin{proof} 
  Note that there are $O(\log \Delta)$ scales and each scale contains 
 (i) $(\log n)^{1/2 + \varepsilon}$ iterations, each of which takes $O(1)$ rounds and 
 (ii) one call to \textsc{GreedyRulingSet} which requires $O(\sqrt{\log n})$ rounds whp by lemma \ref{lemma:pathDelta}. 
 Thus Lines \ref{alg1:scaleloop}-\ref{alg1:endscaleloop} take $O(\log \Delta \cdot (\log n)^{1/2 + \varepsilon})$ rounds.
  From Lemma \ref{lemma:pathM}, the call to the greedy MIS algorithm in Line \ref{alg1:MIS}, takes $O(\frac{\log n}{\log\log n})$ rounds 
  to compute an MIS of each of the $G[M_{i, t}]$'s in parallel.
  This yields the claimed running time.

  Since all nodes in $W_{i,t}$ are at distance 1 from some node in $M_{i,t}$ and we find an MIS of the graph $G[M_{i,t}]$, nodes in 
  $ \cup_i\cup_t (M_{i, t} \cup W_{i,t})$ are all at distance at most $2$ from some node in $I$.
  Nodes that are not in any $M_{i, t} \cup W_{i, t}$ are in $B_t$ for some scale $t$. We compute 2-ruling sets for nodes in $B_t$
  and therefore every node is at most 2 hops from some node in $I$.
\end{proof}}
\subsection{The 3-ruling Set Algorithm}
\label{sec:3-rule-time}
\noindent
This 2-ruling set algorithm can be used to obtain a 3-ruling set algorithm running in $O\left(\frac{\log n}{\log\log n}\right)$ rounds
for \textit{any} graph in the \textsc{Congest} model.
The 3-ruling set algorithm starts by using the simple,
randomized subroutine called \textsc{Sparsify} \cite{BishtKPPODC2014,KP12} to construct in, say $O((\log n)^{2/3})$ rounds,
a set $S$ such that (i) $\Delta(G[S]) = 2^{O((\log n)^{1/3})}$ whp and (ii) every node is in $S$ or has a neighbor in $S$.
\onlyShort{
Using Algorithm \textsc{2-Ruling Set} on $G[S]$ yields the following corollary.
}
\onlyLong{
The properties of \textsc{Sparsify} are more precisely described in the following
lemma.
\begin{lemma}
\label{lemma:sparsify}
(Theorem 1 in \cite{BishtKPPODC2014})
Let $G$ be an $n$-node graph with maximum degree $\Delta$.
Algorithm \textsc{Sparsify} with input $G$ and $f$ runs in $O(\log_f \Delta)$
rounds in the \textsc{Congest} model and produces a set $S \subseteq V(G)$ such that
$\Delta(G[S]) = O(f\cdot \log n)$ whp, and every vertex in $V$ is either
in $S$ or has a neighbor in $S$.
\end{lemma}
}
\onlyLong{
Using Algorithm \textsc{2-Ruling Set} on $G[S]$ yields the following corollary.
}
\begin{corollary}
It is possible to compute a 3-ruling set in $O\left(\frac{\log n}{\log\log n}\right)$ rounds
whp in the \textsc{Congest} model.
\end{corollary}
\onlyLong{
\begin{proof}
First call \textsc{Sparsify} with input graph $G$ and parameter $f = 2^{(\log  n)^{1/3}}$.
\textsc{Sparsify} runs in $O((\log n)^{2/3})$ rounds and returns a set $S$ such that $\Delta(G[S]) = 2^{O((\log n)^{1/3})}$.
We then run \textsc{2-RulingSet} on $G[S]$ with $\varepsilon < 1/6$, which completes in $O\left(\frac{\log n}{\log\log n}\right)$ rounds, returning a $3$-ruling set of $G$.
\end{proof}}
\onlyShort{
  \noindent
  If we relax the problem a bit further and settle for a 5-ruling set, we can improve the round complexity even
further.  In the long version of the paper we present a 5-ruling set algorithm that runs in $O(\sqrt{\log n})$ rounds in 
the \textsc{Congest} model.
This algorithm critically uses the first phase of Ghaffari's MIS algorithm \cite{GhaffariSODA16} and certain
independence properties of nodes that remain undecided after the execution of this phase.}
\onlyLong{
\subsection{5-ruling sets in $O(\sqrt{\log n})$ rounds}
\label{section:5RulingSets}
It turns out that for slightly larger but constant $\beta$, it is possible to compute $\beta$-ruling sets in the
\textsc{Congest} model in $O(\sqrt{\log n})$ rounds. We show this in this section for $\beta = 5$.
The 5-ruling set algorithm (described in pseudocode below) starts by calling the \textsc{Sparsify} 
\cite{BishtKPPODC2014,KP12} subroutine.

\RestyleAlgo{boxruled}
\begin{algorithm2e}\caption{\textsc{5-RulingSet}(Graph $G = (V, E)$):}
      $S  \leftarrow  \textsc{Sparsify}(G,2^{\sqrt{\log n}} )$\;
      $I \leftarrow $ \textsc{GhaffariMISPhase1}$(G[S])$ for $\Theta(c \cdot \log \Delta(G[S]))$ rounds\;
      $R \leftarrow S \setminus (I \cup N(I))$\;
      $I \leftarrow I \cup \textsc{GreedyRulingSet}(G[S], R, 4)$\;
      \textbf{return} $I$\;
\end{algorithm2e}

\noindent
We use $f = 2^{\sqrt{\log n}}$ in our call to \textsc{Sparsify}, which implies that \textsc{Sparsify} runs in $O(\sqrt{\log n})$ rounds.
We then run $O(\log \Delta(G[S]))$ iterations of the first phase of Ghaffari's MIS algorithm on
$G[S]$ and this returns an independent set $I$ \cite{GhaffariSODA16}.
Since $\Delta(G[S]) = 2^{O(\sqrt{\log n})}$, this is equivalent to running $O(\sqrt{log n})$ 
iterations of the first phase of Ghaffari's MIS algorithm.
Recall that the first phase of Ghaffari's algorithm is a ``Luby-like'' algorithm
that runs in the \textsc{Congest} model. 
Finally, we consider the set of nodes that are still undecided, i.e., nodes in $S$ that are 
not in $I \cup N(I)$ and we call Algorithm \textsc{GreedyRulingSet} algorithm with $\beta = 4$ 
in order to compute a 4-ruling set of the as-yet-undecided nodes.
The fact that we compute a 4-ruling set in this step, rather than a $\beta$-ruling set for some $\beta < 4$,
is because of independence properties of Ghaffari's MIS algorithm.
Steps (1)-(3) of \textsc{5-RulingSet} run in $O(\sqrt{\log n})$ rounds either by design or due to properties
of \textsc{Sparsify}. The fact that the greedy 4-ruling set algorithm (Step (4)) also terminates 
in $O(\sqrt{\log n})$ rounds remains to be shown and this partly depends on the
following property of the first phase of Ghaffari's MIS algorithm.
\begin{lemma} 
\label{lemma:Ghaffari}
(Lemma 4.1 in \cite{GhaffariSODA16}) For any constant $c > 0$, for any
set $S$ of nodes that are at pairwise distance at least 5 from each other, the probability
that all nodes in $S$ remain undecided after $\Theta(c \log \Delta)$ rounds of the
first phase of the MIS algorithm is at most $\Delta^{-c|S|}$.
\end{lemma}
\noindent
The sparsification step performed by the call to \textsc{Sparsify} in Step (1) along with Lemma \ref{lemma:Ghaffari} and 
Properties (i) and (ii) of witness paths are used in the following theorem to prove that Step (4) of 
\textsc{5-RulingSet} also completes in $O(\sqrt{\log n})$ rounds whp.
The fact that we are greedily computing a 4-ruling set in Step (4) and any two nodes selected to be in the ruling set
are at least 5 hops away from each other provides the independence that is needed to apply Lemma \ref{lemma:Ghaffari}.

\begin{theorem}
\label{theorem:5RulingSet}
Algorithm \textsc{5-RulingSet} computes a 5-ruling set of $G$ in $O(\sqrt{\log n})$ rounds whp.
\end{theorem}
\begin{proof} 
  \textsc{Sparsify} with paramter $2^{\sqrt{\log n}}$ takes $O(\log_{2^{\sqrt{\log n}}} n) = O(\sqrt{\log n})$ time. 
Also, whp the maximum degree of the graph $G[S]$, $\Delta(G[S])=O(2^{\sqrt{\log n}}\cdot \log n)$
and this is bounded above by $2^{c'\sqrt{\log n}}$ for some constant $c'$.
	Next, we run the \textsc{GhaffariMISPhase1} algorithm for $\Theta(c\log\Delta(G[S])) = O(\sqrt{\log n})$ rounds. 
We now argue that Step (4) also runs in $O(\sqrt{\log n})$ rounds whp. 

Suppose that the call to $\textsc{GreedyRulingSet}(G[S], R, 4)$ takes $p$ iterations, Then, there is a witness path $P = (v_1, v_2, \ldots, v_p)$ 
in $G[S]$ to the execution of this algorithm.
Then, by Lemma \ref{lemma:Ghaffari} and Property (i) of witness paths:
$$\Pr(\text{All nodes in }P \text{ remain undecided after Step (2)}) \leq \left(\frac{1}{2^{c'\sqrt{\log n}}}\right)^{cp}.$$
Now we use Property (ii) of witness paths to upper bound the total number of possible length-$p$ witness paths in $G[S]$ by $|S| \cdot \Delta(G[S])^{5p}$.
Let $E_P$ denote the event that all nodes of a possibe length-$p$ witness path $P$ in $G[S]$ have remained undecided
after Step (2). Taking a union bound over all possible length-$p$ witness paths in $G[S]$ we see that
$$\Pr(\exists\text{ a length-$p$ witness path $P$ in $G[S]$:}\, E_P) \leq n\left(2^{c'\sqrt{\log n}}\right)^{5p} \cdot \left(\frac{1}{2^{c'\sqrt{\log n}}}\right)^{cp} \leq n\left(2^{\sqrt{\log n}}\right)^{pc'(5-c)}.$$
Plugging in $p = \sqrt{\log n}$ and choosing $c \ge 5 + 2/c'$  we get an upper bound of $1/n$ on the probability that a length-$p$ witness
path exists after Step (2).
This implies that whp Step (4) takes $O(\sqrt{\log n})$ rounds.

Step (4) computes a 4-ruling set of $G[R]$ and this along with the set $I$ form a 4-ruling set of $G[S]$.
Since every node is at most 1 hop away from some node in $S$, we have computed a 5-ruling set of $G$.
\end{proof}
}

\section{A Message-Efficient Algorithm for 2-Ruling Set}
\label{sec:msg}
In this section, we present a  randomized  distributed algorithm for computing a 2-ruling set in the \textsc{Congest} model that
 takes $O(n \log^2 n )$ messages and $O(\Delta \log n)$ rounds whp, where $n$ is the number of nodes
and $\Delta$ is the maximum node degree. The algorithm does not require any global knowledge, including knowledge of $n$ or $\Delta$.
We show in \Cref{thm:msg-lb-rulingset} that the algorithm is essentially message-optimal (up to a $\polylog(n)$ factor). 
This is the first message-efficient algorithm known for 2-ruling set, i.e.,
it takes $o(m)$ messages, where $m$ is the number of edges in the graph.\onlyShort{\footnote{In the full paper we present a simpler algorithm for 2-ruling set
that, whp, takes $O(n^{1.5} \log n)$ messages and runs in $O(\log n)$ rounds.}} 
\onlyLong{
In contrast, we show in \Cref{thm:msg-lb-mis} that computing a MIS requires $\Omega(n^2)$ messages (regardless of the number of rounds).
Thus there is a fundamental separation of message complexity between 1-ruling set (MIS) and 2-ruling set computation.
}
\subsection{The Algorithm}
\label{sec:2-msg-algo}
Algorithm~\ref{alg:2-ruling-set}  is inspired by Luby's algorithm for MIS  \cite{LubySICOMP86}; however, there are crucial differences. (Note that
Luby's algorithm sends $\Theta(m)$ messages.) 
Given a ruling set  $R$, we classify nodes in $V$ into three categories: 
\begin{compactitem}
\item $\catone$:  nodes that belong to  the ruling set $R$;
\item $\cattwo$: nodes that have a neighbor  in $R$; and 
\item $\catthree$:  the rest of the nodes, i.e., nodes  that  have a neighbor in $\cattwo$. 
\end{compactitem}
At the beginning of the algorithm,  each node
is $\undecided$, i.e., its category is not set and upon termination, each node knows its category. 

Let us describe one iteration of the algorithm (Steps \ref{alg:msg:begin-loop}-\ref{alg:msg:end-loop}) from the perspective of an arbitrary node $v$. Each undecided node  $v$ marks itself with probably $1/2d(v)$.  If $v$ is marked it samples a set of $\Theta(\log (d(v))$ random neighbours and checks  whether any of them belong to $\cattwo$ --- we call this the {\em checking sampling step}.
 If so, then $v$ becomes a $\catthree$ node and is done (i.e., it will never broadcast again, but will continue to answer checking sampling queries, if any, from its neighbors).   Otherwise, 
$v$ performs the {\em broadcast step}, i.e., it communicates with all its neighbors  and checks if there is a marked neighbor that is of equal or higher degree, and if so, it unmarks itself;
else
it enters the ruling set and becomes a $\catone$ node.\footnote{Alternately, if $v$ finds any $\cattwo$ neighbor (that was missed by checking sampling) during broadcast step it becomes a $\catthree$ node and is done. However, this does not give an asymptotic improvement in the  message complexity analysis compared to
the stated algorithm.} Then node $v$ informs all its neighbors about its  $\catone$ status causing them to become $\cattwo$ nodes (if they are not already) and they are done.

A node that does not hear from any of its neighbors knows that it is not a neighbor of any $\catone$ node. 
Note that $\cattwo$ and $\catthree$ nodes do not initiate messages, which is important for keeping the message complexity low. 
Another main idea in reducing messages is the random sampling check of a few neighbors to see whether any of them are $\cattwo$. Although some nodes might send $O(d(v)$)
messages, we show in Section \ref{sec:2-msg-analysis} that most nodes send (and receive) only $O(\log n)$ messages in an amortized sense. Nodes that remain undecided at the end of one iteration continue to the next iteration. It is easy to implement each iteration in a constant number of rounds.

\begin{algorithm2e}[h]
   $\texttt{status}_v = \undecided$\;
   \While{$\texttt{status}_v = \undecided$} {
     \If{$v$ receives a message from a $\catone$ node} {\label{alg:msg:begin-loop}
       Set $\texttt{status}_v = \cattwo$\; \label{alg:msg:status2}
     }
     \textbf{if} $v$ is $\undecided$ \textbf{then} it marks itself with probability $\frac{1}{2d(v)}$ \;
     \If{$v$ is marked} {
       ({\em Checking Sampling Step:}) Sample a set $A_v$ of $4\log(d(v))$ random neighbors independently and uniformly at random (with replacement) \; \label{alg:msg:check-sampling}
       Find the categories of all nodes in $A_v$ by communicating with them\; \label{alg:msg:find-categories}
       \If{any node in $A_v$ is a $\cattwo$ node} {
         Set $\texttt{status}_v = \catthree$\; \label{alg:msg:status-3}
       }
       \Else {
         ({\em (Local) Broadcast Step:}) Send the marked status and  $d(v)$ value to {\em all} neighbors\;\label{alg:msg:local-broadcast}
         If $v$ hears from an equal or higher degree (marked) neighbor   then $v$ unmarks itself\; \label{alg:msg:unmark}
         If $v$ remains marked, set $\texttt{status}_v = \catone$\;
         Announce status to all neighbors\;
       }
     }\label{alg:msg:end-loop}
   }
   \caption{Algorithm {\tt 2-rulingset-msg}: code for a node $v$. $d(v)$ is the degree of $v$.} 
   \label{alg:2-ruling-set}
\end{algorithm2e}

\subsection{Analysis of Algorithm {\tt 2-rulingset-msg}}
\label{sec:2-msg-analysis}

One \emph{phase} of the algorithm consists of Steps \ref{alg:msg:begin-loop}-\ref{alg:msg:end-loop}, which can be implemented in a constant number of rounds.
We say that a node is \emph{decided} if it is in $\catone$, $\cattwo$, or $\catthree$.
The first lemma, which is easy to establish,  shows that if a node is marked, it has a good chance to get decided.

\begin{lemma}
\label{le:mark}
A node that marks itself in any phase gets decided with probability at least $1/2$ in that phase. 
Furthermore, the probability that a node remains undecided after $2\log n$ marked phases is at most $1/n^2$.
\end{lemma}
\onlyLong{
\begin{proof}
Consider a marked node $v$. We only consider the probability that $v$ becomes a $\catone$ node, i.e.,
part of the independent set. (It can also become decided and  become a $\catthree$ or $\cattwo$ node .) 
A marked node becomes unmarked if an equal or higher degree neighbor is marked.
The probability of this ``bad" event happening is at most
$$\sum_{u \in N(v) : d(u) \geq d(v)} \frac{1}{2d(u)} \leq \sum_{u \in N(v) : d(u) \geq d(v)} \frac{1}{2d(v)} \leq \sum_{u \in N(v)} \frac{1}{2d(v)} \leq \frac{d(v)}{2d(v)} = \tfrac{1}{2}.$$ 
The probability that a  node remains undecided after $2\log n$ marked phases is at most $\frac{1}{2^{2\log n}} \leq 1/n^2$.
\end{proof}
}
\noindent
The next lemma bounds the round complexity of the algorithm and establishes its correctness. 
The round complexity bound is essentially a consequence of the previous lemma and the correctness of the 
algorithm is easy to check.

\begin{lemma}
\label{le:runtime}
The  algorithm  {\tt 2-rulingset-msg}  runs in $O(\Delta \log n)$ rounds whp. In particular, with probability at least $1-2/n^2$,
a node $v$ becomes decided after $O(d(v)\log n)$ rounds.
When the algorithm terminates, i.e., when all nodes are decided, the $\catone$-nodes form a $2$-ruling set of the graph. Moreover, each node
is correctly classified according to its category.
\end{lemma}
\onlyLong{
\begin{proof}
Consider a node $v$ with degree $d(v)$. In $16d(v)\log n$ phases, it marks itself $8 \log n$ times in expectation, assuming that it is still undecided.
Using a  Chernoff bound --- lower tail --  (cf. Section \ref{sec:concentration}), it follows that the node is marked at least $2 \log n$ times with probability at least $1- 1/n^2$.
By Lemma \ref{le:mark}, the probability that $v$ is still undecided if it gets marked $2\log n$ times is at most $1/n^2$.
Hence, unconditionally, the probability that a node is still undecided after $16d(v)\log n$ phases is at most $2/n^2$.
Applying a union bound over all nodes, the probability that any node is undecided after $16\Delta \log n$ phases  is at most
$2/n$. Hence the algorithm finishes in $O(\Delta \log n)$ rounds with high probability.

From the description of the algorithm it is clear that when the algorithm ends, every node has entered into either $\catone$, $\cattwo$, or $\catthree$.
By the symmetry breaking step (Step~\ref{alg:msg:unmark}), $\catone$ nodes form an independent set. They also form a ruling set because, $\cattwo$ nodes are neighbors of $\catone$ nodes (Step~\ref{alg:msg:status2}) and a node becomes $\catthree$ if it is not a neighbor of a $\catone$, but 
is a neighbor of a $\cattwo$ node (and hence is at distance 2 from a $\catone$ node). 
\end{proof}
}

\noindent
We now show a technical lemma that is crucially used in proving the message complexity bounds of the algorithm in Lemma \ref{le:msg}.  It gives a high probability bound on  the total number of messages sent by all nodes during the Broadcast step in any particular phase  (i.e., Step \ref{alg:msg:local-broadcast}) of the algorithm in terms of  a quantity that depends on the number of undecided nodes and their neighbors.  While bounding the expectation is easy, showing concentration is more involved. (We note
that we really use only part (b) of the Lemma for our subsequent analysis, but showing part (a) first, helps understand the proof of part (b)).

\onlyLong{We call a node's checking sampling step a ``success'', if it results in finding a $\cattwo$ node (in this case the node will get decided in Step \ref{alg:msg:status-3}), 
otherwise, it is called a ``failure''.}

\begin{lemma}
\label{le:msgbound}
Let $U \subseteq V$ be a (sub-)set of undecided nodes  at the beginning of a phase.
Let $N(v)$ be the set of  neighbors of $v$. 
Then the following statements hold: \\
{\bf (a)} Let $Z(U) = U \cup (\cup_{v \in U}N(v))$. The total number of messages sent by all nodes in $U$ during the Broadcast step in this phase  (i.e., Step \ref{alg:msg:local-broadcast}) of the algorithm is $O(|Z(U)| \log n)$ 
with probability at least $1 - 1/n^3$. 

\noindent
{\bf (b)}  Let $N'(v)$ be the set of undecided and category 3 neighbors of $v$ and suppose $|N'(v)| \geq d(v)/2$ (where $d(v)$ is the degree of $v$), for each $v \in U$.
Let $Z'(U) = U \cup (\cup_{v \in U}N'(v))$. The total number of messages sent by all nodes in $U$ during the Broadcast step in this phase  (i.e., Step \ref{alg:msg:local-broadcast}) of the algorithm is $O(|Z'(U)| \log n)$ 
with probability at least $1 - 1/n^3$. 
\end{lemma}

\onlyLong{
\begin{proof}
A node $v$ enters the broadcast step only if it marks itself and if it does not find any $\cattwo$ neighbor in its checking sampling step.
The marking probability is $1/2d(v)$. Hence (even) ignoring the checking sampling step, the probability that a node broadcasts is at most $1/2d(v)$ (note 
that in the very first phase, checking sampling will result in failure for all nodes).
Let random variable (r.v.) $X_v$ denote the number of messages
broadcast by node $v$ in this phase. Hence, $E[X_v] = \frac{1}{2d(v)}d(v) = 1/2$. Let random variable $X$ denote the total number of messages
broadcast in this phase: $X = \sum_{v \in U} X_v$. By linearity of expectation, 
the  expected number of messages broadcast in one phase is $E[X] = \sum_{v \in U}  E[X_v] = k/2$, where $k = |U|$.
We next show concentration of  $X$.
We note that $X_v$s are all independent and, for the variance of $X_v$, we get 
\begin{align*}
  \var{X_v} &= E[X^2_v] - (E[X_v])^2 =  \frac{1}{2d(v)}(d(v))^2 - \left(\tfrac{1}{2}\right)^2 = \frac{d(v)}{2} - \frac{1}{4} \leq \frac{d(v)}{2}.\\
\intertext{It follows that}
  \var{X} &= \sum_{v \in U} \var{X_v} \leq \sum_{v \in U} d(v)/2 \leq (|Z(U)|^2)/4. 
\end{align*}
  We have $\sum_{v \in U} d(v) \leq (|Z(U)|^2)/2$, since the latter counts all possible edges in the subgraph induced by $U$ and its neighbors. Furthermore, $X_v - E[X_v] \leq d(v)\leq |Z(U)|$. Thus, we can apply Bernstein's inequality (cf. Section \ref{sec:concentration}) to obtain
\begin{align*}
  \Pr(X \geq k +  4|Z(U)| \log n) 
    &= \exp\left(-\frac{16|Z(U)|^2 \log^2 n}{2Var(X)+  (8/3)|Z(U)|^2\log n} \right) \\
    &\le \exp\left(-\frac{16|Z(U)|^2 \log^2 n}{(|Z(U)|^2)/2+  (8/3)|Z(U)|^2\log n}\right),
\end{align*}
which is at most $1/n^2$, completing part (a).

To show part (b), let $N_2(v)$ be the set of $\cattwo$ neighbors of a node $v \in U$. Then $N_2(v) = N(v) - N'(v)$ and  
\begin{equation}
  \label{eq:bound_dv}
  \sum_{v \in U} d(v) \leq  \sum_{v \in U} |N'(v)|  + \sum_{v \in U} |N_2(v)|.
\end{equation}
We will now bound from above the two sums on the right-hand side.
Note that  $\sum_{v \in U} |N'(v)| \leq  |Z'(U)|^2/2 $ since the latter counts all possible edges in the subgraph induced by $U$ and its $\catthree$ and undecided neighbors. 
By assumption, $|N_2(v)| \leq d(v)/2$, which means that
\begin{align*}
  \sum_{v \in U} |N_2(v)|
 \leq  \sum_{v \in U} d(v)/2 \leq k |Z'(U)| \leq |Z'(U)|^2. \\
 \intertext{Plugging these bounds into \eqref{eq:bound_dv}, we get }
   \sum_{v \in U} d(v)  \leq  |Z'(U)|^2/2 +  |Z'(U)|^2 \leq (3/2)  |Z'(U)|^2.
\end{align*}
Hence $\var{X} \leq (3/4)  |Z'(U)|^2$. Let $X_v$ and $X$ be defined as above. We have, $X_v - E[X_v] \leq d(v)\leq 2|Z'(U)|$. Now,  applying Bernstein's inequality shows a similar concentration bound for $X$ as in part (a).
\end{proof}
}

\begin{lemma}
\label{le:msg}
The  algorithm  {\tt 2-rulingset-msg} uses $O(n \log^2 n)$ messages whp.
\end{lemma}
\begin{proof}
We will argue separately about two kinds of messages that any node can initiate. Consider any node $v$. 
\noindent
{\bf 1.}  {\em type 1 messages}:  In the checking sampling step in some phase, $v$ samples $4\log d(v)$ random neighbours which costs $O(\log d(v))$ messages in that phase.

\noindent
{\bf 2.} {\em type 2 messages}:  In the broadcast step in some phase, $v$ sends to all its neighbors  which costs $d(v)$ messages.  This happens when all the sampled neighbors in set $A_v$ (found in Step \ref{alg:msg:find-categories})  are not  $\cattwo$ nodes.

Note that $v$ initiates any message at all, i.e., both  type 1 and 2 messages happen, only when  $v$ marks itself, which happens with probability $1/2d(v)$.

We first bound the  type 1 messages sent overall by all nodes. By the above statement, a node does checking sampling
when it marks itself which happens with probability $1/2d(v)$. By Lemma \ref{le:mark},  with probability at least $1 - 1/n^2$,  a node is marked (before it gets decided) at most $2 \log n$ times. 
Hence, with probability at least $1 - 1/n^2$, the number of type 1 messages sent by  node $v$ is at most $O(\log d(v) \log n)$; this implies, 
by union bound, that with probability at least $1 - 1/n$ every node $v$ sends at most $O(\log d(v) \log n)$ type 1 messages.
Thus, whp, the total number of type 1 messages sent is $\sum_{v \in V} O(\log d(v) \log n) = O(n \log^2 n)$.

We next bound the type 2 messages, i.e., messages sent during the broadcast step. There are two cases to consider in any phase.  

\paragraph{Case 1} In this case we focus (only)  on the broadcast messages of the set $U$ of   undecided nodes $v$  that (each) have at least $d(v)/2$ neighbors that are in $\catthree$ or undecided (in that phase).
We show by a charging argument that any node  receives amortized $O(\log n)$ messages (whp) in this case. 
When a node $u$ (in this case)  broadcasts,
its $d(u)$ messages are charged equally to itself and its $\catthree$ and undecided neighbors (which number at least $d(u)/2$).

We first show that any $\catthree$ or undecided node $v$  is charged by amortized $O(\log n)$ messages  in any  phase.  
Consider the set $U(v)$ which is the set of undecided  nodes (each of which  satisfy Case 1 property of having 
at least half of its neighbors that are in $\catthree$ or undecided in this phase)  in the closed neighborhood of $v$ 
(i.e., $\{v\} \cup N(v)$). 
As in Lemma \ref{le:msgbound}.(b), we define $Z'(U(v)) =  U(v) \cup (\cup_{w \in U(v)}N'(w))$, where $N'(w)$ is the set of all $\undecided$ or $\catthree$ neighbors of $w$. 
 Since, by assumption of Case 1, every undecided node $u \in U(v)$ has at least $d(u)/2$ neighbors that are in $\catthree$ or $\undecided$ in the current phase, applying  Lemma \ref{le:msgbound} (part (b)) to the set $Z'(U(v))$ tells us that, with probability at least $1 -1/n^2$, the total number of messages broadcast by undecided nodes in $U(v)$  is $O(|Z'(U(v))| \log n)$. Hence,
amortizing over the total number of (undecided and $\catthree$) nodes in $Z'(U(v))$,  we have shown each node in $Z'(U(v))$,
 in particular $v$,  is charged (amortized) $O(\log n)$  in a phase.  Taking a union bound, gives a high probability result for all nodes $v$.

To show that the same node $v$ is not charged too many times {\em across phases}, we use the fact that $\cattwo$ nodes are never charged (and they do not broadcast). 
We  note that if a node enters the ruling set (i.e., becomes $\catone$) in some phase, then all its neighbors 
 become $\cattwo$ nodes and will never be charged again (in any subsequent phase).  Furthermore, since a marked node enters the ruling set  with probability at least $1/2$, a neighbor of $v$ (or $v$ itself) gets charged at most $O(\log n)$ times whp.  Hence overall a node is charged at most $O(\log^2 n)$ times whp and by union bound, every node gets charged at most  $O(\log^2 n)$ times whp.

\paragraph{Case 2} In this case, we focus on the messages broadcast by those undecided nodes $v$ that have at most $d(v)/2-1$ neighbors that are in $\catthree$ or undecided, i.e., at least $d(v)/2+1$ neighbors are in $\cattwo$. By the description of our algorithm, a node enters the broadcast step, only if checking sampling step (Step \ref{alg:msg:check-sampling}) fails to find a $\cattwo$ node. The probability of this ``bad'' event happening is at most $\frac{1}{d(v)^4}$, which is the probability that 
a $\cattwo$ neighbor (of which there are at least $d(v)/2$ many) is not among any of the $4 \log (d(v))$ randomly sampled neighbors.
We next bound the total number of broadcast messages generated by all undecided nodes in Case 2 during the entire course of the algorithm.
By Lemma \ref{le:mark}, for any node $v$, Case 2 can potentially happen only $2\log n$ times with probability at least $1-1/n^2$, since  that is the number of times $v$ can get marked. Let r.v. $Y_v$ denote the number of Case 2 broadcast messages sent by $v$
during the course of the algorithm. Conditional on the fact that it gets marked at most $2 \log n$ times, we have
$E[Y_v] = 2\log n \frac{1}{d(v)^4}d(v) = 2\log n \frac{1}{d(v)^3}$.

Let $Y = \sum_{v \in V} Y_v$. Hence, conditional on the fact that each node gets marked at most $2\log n$ times (which happens with probability $\geq 1 -1 /n$) the total expected number of Case 2 broadcast messages 
sent by all nodes is 
 $E[Y] = \sum_{v \in V} E[Y_v] = \sum_{v \in V} 2\log n \frac{1}{d(v)^3} = O(n\log n).$

We next show concentration of $Y$ (conditionally as mentioned above). We know that
$
  \var{Y_v} = 4\log^2 n (\frac{1}{d(v)^2} - \frac{1}{d(v)^6}) \leq 4 \log^2 n. $
Since the random  variables $Y_v$ are independent, we have 
$\var{Y} = \sum_{v\in V} Var (Y_v) = 4n\log^2 n$.
Noting that $Y_v -E[Y_v] \leq 2n\log n$, we apply Bernstein's inequality to obtain
\[
  \Pr(Y \geq E[Y] + 4n\log^2 n) \leq \exp\left(-\frac{16n^2\log^4n}{8n\log^2 n + (2/3)2n\log n(4n\log^2 n)}\right) \leq O(1/n^2).
\]
Since the conditioning with  respect to the fact that all nodes get marked at most $ 2 \log n$ times happens with probability
at least $1-1/n$, unconditionally, $\Pr(Y \geq \Theta(n\log^2 n)) \leq O(1/n^2) + 1/n$.
Hence, the overall broadcast messages sent by nodes in Case 2 is bounded by $O(n\log^2 n)$ whp.

Combining type 1 and type 2 messages, the overall number of messages is bounded by $O(n \log^2 n)$ whp.
\end{proof}

\noindent
Thus we obtain the following theorem.
\onlyShort{In the full paper, we show that this analysis of the Algorithm {\tt 2-rulingset-msg} is tight up to a 
polylogarithmic factor.}
\begin{theorem}  
 \label{thm:2-ruling-set-msg}
The  algorithm  {\tt 2-rulingset-msg} computes a 2-ruling set using $O(n \log^2 n)$ messages and terminates in $O(\Delta \log n)$ rounds with high probability.
\end{theorem}

\onlyLong{
\paragraph{A tight example} We show that the above analysis of the Algorithm {\tt 2-rulingset-msg}  is tight up to a polylogarithmic factor. The tightness of the message complexity follows from \Cref{thm:msg-lb-rulingset} which shows that any $O(1)$-ruling set algorithm
needs $\Omega(n)$ messages.

For the time complexity, we show that the analysis is essentially tight by giving a $n$-node graph where the algorithm takes $\Omega(n^{1-\epsilon})$ rounds in a graph where $\Delta = O(n)$, for any fixed constant $\epsilon > 0$. The graph is constructed as follows. 
The graph consists of a distinguished node $s$ and  three sets of nodes --- sets $A, B,$ and $C$. $A$ has $n^{1-\epsilon}$ nodes, $B$ has $n^{1-\epsilon'}$ nodes,
where $\epsilon > \epsilon' > 0$ and $C$ has $n - 1 - n^{1-\epsilon} - n^{1-\epsilon'}$ nodes. There is a complete bipartite graph between sets $A$ and $B$ and between sets $B$ and $C$
and $s$ is connected to all nodes in $A$. Hence $s$ has degree $n^{1-\epsilon}$, every  node in $A$ has degree $\Theta(n^{1-\epsilon'})$, every node in $B$ has degree $\Theta(n)$
and every  node in $C$ has degree $\Theta(n^{1-\epsilon'})$. 
If we run the algorithm {\tt 2-rulingset-msg} on this graph, then with high probability at least one node in $C$ will enter the ruling set in the first phase itself; further, with probability at least $1-o(1)$, neither $s$ nor any node in sets $A$ and $B$ mark themselves in the first phase.
Once a node from $C$ enters the ruling set, all nodes in $B$ become $\cattwo$ nodes. On the other hand, nodes in sets $A$ and $C$ (conditioned on not entering the ruling set in the current phase) will become $\catthree$ nodes in the next phase by executing the checking sampling step (which will succeed with high probability). However, since all of the neighbors of $s$ are in $\catthree$, node $s$ is bound to execute $\Theta(n^{1-\epsilon})$ phases until it marks itself and enter the ruling set in expectation. Hence the expected round complexity is $\Omega(n^{1-\epsilon})$. Note that even though this graph has $\Theta(n^{2-\epsilon'})$ edges, the algorithm sends only $O(n \log^2 n)$ messages (whp).
 }
 
\onlyLong{
\subsection{An $O(\log n)$-round, $O(n^{1.5}\log n)$ message complexity 2-ruling set algorithm}
\label{sec:fastmsg}
We show that when $m$ is large,  one can design a simple algorithm with $o(m)$  message complexity algorithm that runs in time $O(\log n)$.
This algorithm requires knowledge of $n$.
The algorithm is as follows.

\begin{enumerate}
\item Initially all nodes are {\em inactive}.
\item  Every node with degree less than $\sqrt{n}$  becomes active  and nodes with degree higher than $\sqrt{n}$  become active independently  with probability  $2\log n/\sqrt{n}$. Let $S$ denote the set of active nodes.
\item Nodes in $S$ broadcast their status to all nodes (thus each active node knows its active neighbors, if any).
\item Compute, using Luby's algorithm, an MIS of $S$ and return it.
\end{enumerate}

\begin{theorem}
 \label{thm:2-ruling-set-msg-fast}
The above algorithm computes a 2-ruling set using $O(n^{1.5} \log n)$ messages and $O(\log n)$ rounds.
\end{theorem}

\begin{proof}
We show that the set returned, i.e., the MIS of $S$, is a 2-ruling set.
Every node in $S$ (the active set) is either in the  MIS or a neighbor of a node in the MIS. 
We next show that every node in $v\in V-S$ has a neighbor in $S$.
If $v$  has a neighbor $u$ of degree less than $\sqrt{n}$, then $u$ belongs to $S$. Otherwise, since $v$'s degree is at least $\sqrt{n}$,
the probability that at least one of its neighbors (all of which must have a degree at least $\sqrt{n}$) belonging to $S$ 
is at least $1 - (1 - 2 \log n/\sqrt{n})^{\sqrt{n}} \geq 1 - e^{-2\log n} = 1- 1/n^2$. Hence, by a union bound, every node in $V-S$ has a neighbor in $S$ whp.
It follows that an MIS of $S$ is a 2-ruling set of the graph.

We next analyze the time and message complexity. The time complexity follows immediately from the run time of Luby's algorithm.
We can show that, whp, the number of messages in Step (3)  is $O(n^{1.5}\log n)$ as follows. Nodes with degree
less than $\sqrt{n}$ contribute $O(n^{1.5})$ messages. For nodes that have a degree higher than $\sqrt{n}$,
the expected number of neighbors in $S$ is $O(\sqrt{n}\log n)$ and this holds whp (by applying a standard Chernoff bound). Hence the 
number of messages is $O(n^{1.5}\log n)$ whp.
From the above, it follows that the sum of the degrees of the nodes in $S$  is bounded by $O(n^{1.5}\log n)$. Hence, we observe that Step (4) --- Luby's algorithm on $S$ --- requires $O(n^{1.5})$ messages. 
Hence, overall the total message complexity is $O(n^{1.5}\log n)$.
\end{proof}
}

\section{Message Complexity Lower Bounds}
\label{sec:lower}
We first point out that the bound of Theorem \ref{thm:2-ruling-set-msg} is tight up to logarithmic factors.
\onlyShort{The proof is a simple indistinguishability argument and is relegated to the full paper.}

\begin{theorem} \label{thm:msg-lb-rulingset}
Any $O(1)$-ruling set algorithm that succeeds with probability $1-o(1)$ sends $\Omega(n)$ messages in the worst case.  
This is true even if nodes have prior knowledge of the network size $n$.
\end{theorem}
\onlyLong{%
\begin{proof}
Consider a cycle $G$ of $n$ nodes and, for any $t=O(1)$, suppose that $\cA$ is a $t$-ruling set algorithm that has a message complexity of $o(n)$. 

We first condition on the assumption that nodes do not have IDs and subsequently remove this restriction. 
By assumption, there are $\ge (1 - \frac{1}{2t+1})n$ nodes in any run of $\cA$ that are quiescent, i.e., neither send nor receive any messages. 
We define a \emph{segment} to be a sequence of consecutive nodes in the cycle. 
Let $E_u$ be the indicator random variable that is $1$ if and only if node $u$ enters the ruling set. 
It follows that there exists a segment $S$ of $2t+1$ quiescent nodes and
we can observe that the random variables in the set $\{ E_v \mid v \in S \}$ are independent. 

Consider a sub-segment of $3$ nodes in $S$.  
Since the network is anonymous, we know that $\Prob{ E_u \!=\! 1 } = \Prob{ E_v \!=\! 1} = p$, for any $u,v \in S$.
Recalling that $\cA$ succeeds with probability $1 - o(1)$ tells us that the event where two neighbors join the ruling set happens with probability at most $2p^2(1 - p) = o(1)$ and hence it must be that $p = o(1)$.
On the other hand, since $S$ has length $2t+1$, at least $1$ node must enter the ruling set with probability $\ge 1 - o(1)$ and conversely $(1 - p)^{2t+1} = o(1)$, which contradicts $t = O(1)$ and $p = o(1)$.

Finally, suppose that an algorithm $\cB$ computes a $t$-ruling set correctly with probability $1 - o(1)$ requiring $o(n)$ messages in networks where nodes have unique IDs.
Similarly to \cite{LE-tcs}, we construct an algorithm $\cA$ that works in anonymous networks by instructing every node to first uniformly at random choose a unique ID from $[1,n^c]$, where $c\ge 4$ is any constant, and then run algorithm $\cB$ with the random ID as input. With high probability, all randomly chosen IDs are unique and hence $\cA$ also succeeds with probability $1 - o(1)$, contradicting the lower bound above.
\end{proof}
}
\noindent
Next, we show a clear separation between the message complexity of computing an $t$-ruling set ($t>1$) and a maximal independent set (i.e., 1-ruling set) by proving an unconditional $\Omega(n^2)$ lower bound for the latter. 
\begin{theorem} \label{thm:msg-lb-mis}
Any maximal independent set algorithm that succeeds with probability $1 - \epsilon$  on connected networks, where $0\le\epsilon<\tfrac{1}{2}$ is a constant, must send $\Omega(n^2)$ messages in the worst case.
This is true even if nodes have prior knowledge of the network size $n$.
\end{theorem}

\noindent\emph{Proof of \Cref{thm:msg-lb-mis}.}
For the sake of a contradiction, assume that there is an algorithm $\cA$ that sends at most $\mu = o(n^2)$ messages in the worst case and succeeds with probability $\ge 1 - \epsilon$, for some $\epsilon < \tfrac{1}{2}$. 

Consider two copies $G$ and $G'$ of the complete bipartite graph on $n/2$ nodes.\footnote{To simplify our analysis, we assume that $n/2$ and $n/4$ are integers.}
For now, we consider the anonymous case where nodes do not have access to unique IDs; we will later show how to remove this restriction.
Recall that in our model (cf.\ \Cref{sec:model}), we assume that nodes do no have any prior knowledge of their neighbors in the graph.
Instead, each node $u$ has a list ports $1,\dots,deg_u$, whose destination are wired in advance by an adversary. 

We consider two concrete instances of our lower bound network depending on the wiring of the edges.
First, let $D=(G,G')$ be the disconnected graph consisting of $G$ and $G'$ and their induced edge sets. 
It is easy to see that there are exactly $4$ possible choices for an MIS on $D$,
as any valid MIS must contain the entire left (resp.\ right) half of the nodes in $G$ and $G'$ and no other nodes.
We denote the events of obtaining one of the four possible MISs by $LL'$, $LR'$, $RL'$, $RR'$, where, e.g., $RL'$ is the event that the right half of $G$ (i.e.\ nodes in $R$) and the left half of $G'$ (i.e.\ nodes in $L'$) are chosen.
Let ``on $D$'' be the event that $\cA$ is executed on graph $D$.
Of course, we cannot assume that algorithm $\cA$ does anything useful on this graph as we require $\cA$ only to succeed on connected networks.
However, we will make use of the symmetry of the components of $D$ later on in the proof.

\begin{observation} \label{obs:sameprob}
  $\Prob{ LL' \mid \text{on $D$} } = 
  \Prob{ LR' \mid \text{on $D$} } = 
  \Prob{ RL' \mid \text{on $D$} } = 
  \Prob{ RR' \mid \text{on $D$} }$.
\end{observation}

\noindent
We now define the second instance of our lower bound graph.
Consider any pair of edges $e=(u,v) \in G=(L,R)$ and $e'=(u',v') \in G'=(L',R')$.
We define the \emph{bridge graph} by removing $e$ and $e'$ from $G$ respectively $G'$ and, instead, adding the \emph{bridge edges} $b = (u,u')$ and $b' = (v,v')$ by connecting the same ports that were used for $e$ and $e'$; see \Cref{fig:messagelb}.
We use $B$ to denote a graph that is chosen uniformly at random from all possible bridge graphs, i.e., the edges replaced by bridge edges are chosen uniformly at random according to the above construction. 
Let ``\text{$G\leftrightarrow G'$}'' be the event that $\cA$ sends at least $1$ message   
over a bridge edge and, similarly, we use ``\text{$G\not\leftrightarrow G'$}'' to denote the event that this does not happen.

\begin{lemma} \label{lem:bc}
Consider an execution of algorithm $\cA$ on a uniformly at random chosen bridge graph $B$.
The probability that a message is sent across a bridge is $o(1)$, i.e., $\Prob{\text{$G\not\leftrightarrow G'$}} = 1 - o(1)$.
\end{lemma}
\onlyLong{
\begin{proof}
We start out by observing that nodes do not have any prior knowledge regarding the choice of the bridge edges when $\cA$ is executed on graph $B$.
Moreover, as long as no message was sent across a bridge edge, every unused port (i.e., over which the algorithm has not sent a message yet) has the same probability of being connected to a bridge edge. 
It follows that discovering a bridge edge corresponds to sampling without replacement, which we can model using the hypergeometric distribution, where we have exactly $4$ bridge ports among the edge set of $B$ and one draw per message sent by the algorithm.
Observe that the total number of ports is $2|E(B)| = 4\left(\tfrac{n}{4}\right)^2 = \tfrac{n^2}{4}$. Hence, by the properties of the hypergeometric distribution, the expected number of bridge edges discovered when the algorithm sends at most $\mu$ messages is  
$\Theta(\frac{\mu}{n^2})$.
Using Markov's Inequality and the fact that $\mu = o(n^2)$, it follows that the probability of discovering at least $1$ bridge edge is $o(1)$.
\end{proof}
}

A crucial property of our construction is that, as long as no bridge edge is discovered, the algorithm behaves the same on $B$ as it does on $D$.
The following lemma can be shown by induction over the number of rounds. 

\begin{lemma} \label{lem:sameview}
  Let $Y$ be any event that is a function of the communication and computation performed by algorithm $\cA$.
  Then, $\Prob{ Y \mid \text{$G\not\leftrightarrow G'$}} = 
         \Prob{ Y \mid \text{on $D$}}$.
\end{lemma}


We are now ready to prove \Cref{thm:msg-lb-mis}.
Consider a run of algorithm $\cA$ on a uniformly at random chosen bridge graph $B$. 
Let ``$\cA$ succ.'' denote the event that $\cA$ correctly outputs an MIS.
Observe that $\cA$ succeeds when executed on $B$ if and only if we arrive at an output configuration corresponding $LR'$ or $RL'$.
It follows that
\begin{align*}
    \Prob{\text{$\cA$ succ.} } = \!\!\!
  \sum_{W \in \{LR',RL'\}}\!\!\!\! \Prob{ W \mid \text{$G\not\leftrightarrow G'$ }}\cdot \Prob{\text{$G\not\leftrightarrow G'$}}
    + \Prob{\text{$\cA$ succ.} \mid \text{$G\leftrightarrow G'$}} \cdot
\Prob{\text{$G\leftrightarrow G'$}}
\ge 1 - \epsilon.
\end{align*}
{\Cref{lem:bc} tells us that $\Prob{\text{$G\leftrightarrow G'$}}=o(1)$ and, using  $\Prob{\text{$G\not\leftrightarrow G'$}}\le 1$, allows us to rewrite the above inequality as}
\onlyLong{
\begin{align*}
  \sum_{W \in \{LR',RL'\}}\! \Prob{ W \mid \text{$G\not\leftrightarrow G'$ }} 
\ge 1 - \epsilon - o(1).
\end{align*}
}
\onlyShort{$\sum_{W \in \{LR',RL'\}}\! \Prob{ W \mid \text{$G\not\leftrightarrow G'$ }}
\ge 1 - \epsilon - o(1).$}
Applying \Cref{lem:sameview} to the terms in the sum, we get
\onlyLong{
\begin{align}
  \sum_{W \in \{LR',RL'\}}\! \Prob{ W \mid \text{on $D$}} \label{eq:epsilon}
\ge 1 - \epsilon - o(1).
\end{align}
}
\onlyShort{$\sum_{W \in \{LR',RL'\}}\! \Prob{ W \mid \text{on $D$}} \label{eq:epsilon}
\ge 1 - \epsilon - o(1).$}
By Observation~\ref{obs:sameprob}, we know that
\onlyLong{
  \[ 
  \Prob{ LR' \mid \text{on $D$}} + \Prob{ RL' \mid \text{on $D$}} \le \tfrac{1}{2},
\]
}
\onlyShort{
  $\Prob{ LR' \mid \text{on $D$}} + \Prob{ RL \mid \text{on $D$}} \le \tfrac{1}{2},$
}
which we can plug into\onlyLong{ \eqref{eq:epsilon}}\onlyShort{ the previously obtained bound on $\sum_{W \in \{LR',RL'\}}\! \Prob{ W \mid \text{on $D$}}$} to obtain $\epsilon \ge \tfrac{1}{2} - o(1)$, yielding a contradiction. 

Finally, we can remove the restriction of not having unique IDs by arguing that the algorithm can generate unique IDs with high probability, since we assume that nodes know $n$; see the proof of \Cref{thm:msg-lb-rulingset}\onlyShort{ in the full paper} for a similar argument.
This completes the proof of \Cref{thm:msg-lb-mis}.
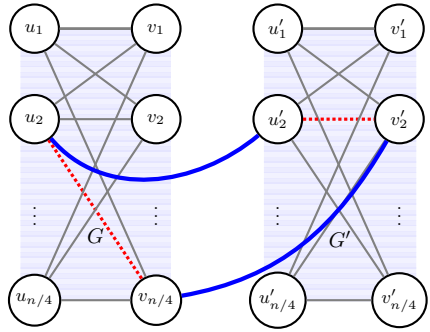
\begin{figure}[t]
\begin{center}
\pgfdeclarelayer{background}
\pgfdeclarelayer{inbetween}
\pgfsetlayers{background,inbetween,main}
\begin{tikzpicture}[node distance=1.5cm,scale=0.8, every node/.style={scale=0.8}]
\tikzstyle{v}=[circle,draw=black,fill=white!30,thick,inner sep=2pt,minimum size=8mm]
\tikzstyle{link}=[-,black!50,thick,auto]
\tikzstyle{arc}=[draw=black,rectangle,rounded corners,fill=lightgray,font=\normalsize]
\tikzstyle{arcNeighborhoodZ}=[semitransparent,thick,draw=none,rectangle,rounded corners,fill=blue!20,pattern=horizontal lines light blue]
\tikzstyle{arcNeighborhoodZ1}=[arcNeighborhoodZ,fill=red!20,dashed]
\footnotesize 
\node[v]  (u001)                  {$u_{1}$};
\node[v]  (u002)  [below of=u001] {$u_{2}$};
\node[]  (udots1)[below of=u002] {$\vdots$};
\node[v]  (u00n4) [below of=udots1] {$u_{n/4}$};

\node[v]  (v001)  [xshift=2cm]{$v_{1}$};
\node[v]  (v002)  [below of=v001] {$v_{2}$};
\node[]  (vdots1)[below of=v002] {$\vdots$};
\node[v]  (v00n4) [below of=vdots1] {$v_{n/4}$};

\node[v]  (u001d)  [xshift=4cm]{$u_{1}'$};
\node[v]  (u002d)  [below of=u001d] {$u_{2}'$};
\node[]  (udots1d)[below of=u002d] {$\vdots$};
\node[v]  (u00n4d) [below of=udots1d] {$u_{n/4}'$};

\node[v]  (v001d)  [xshift=6cm]{$v_{1}'$};
\node[v]  (v002d)  [below of=v001d] {$v_{2}'$};
\node[]  (vdots1d)[below of=v002d] {$\vdots$};
\node[v]  (v00n4d) [below of=vdots1d] {$v_{n/4}'$};

\draw[link] (u001) to (v001);
\draw[link] (u001) to (v002);
\draw[link] (u001) to (v00n4);
\draw[link] (u001) to (v001);
\draw[link] (u001) to (v001);
\draw[link] (u002) to (v001);
\draw[link] (u002) to (v002);
\draw[link,very thick,densely dotted,red] (u002) to (v00n4);
\draw[link] (u00n4) to (v001);
\draw[link] (u00n4) to (v002);
\draw[link] (u00n4) to (v00n4);

\draw[link] (u001d) to (v001d);
\draw[link] (u001d) to (v002d);
\draw[link] (u001d) to (v00n4d);
\draw[link] (u001d) to (v001d);
\draw[link] (u001d) to (v001d);
\draw[link] (u002d) to (v001d);
\draw[link,very thick,densely dotted,red] (u002d) to (v002d);
\draw[link] (u002d) to (v00n4d);
\draw[link] (u00n4d) to (v001d);
\draw[link] (u00n4d) to (v002d);
\draw[link] (u00n4d) to (v00n4d);

\draw[link,blue,ultra thick,in=240,out=10] (v00n4) to (v002d);
\draw[link,blue,ultra thick,in=220,out=310] (u002) to (u002d);

\begin{pgfonlayer}{background}
  \node[arcNeighborhoodZ,fit=(u001) (v00n4),label={[yshift=1.30cm,font=\normalsize]below:$G$}] (C0) {};
  \node[arcNeighborhoodZ,fit=(u001d) (v00n4d),label={[yshift=1.30cm,font=\normalsize]below:$G'$}] (C0) {};
\end{pgfonlayer}
\begin{pgfonlayer}{inbetween}
\end{pgfonlayer}
\end{tikzpicture}
\end{center}
\caption{The lower bound graph $B(G,G')$ for \Cref{thm:msg-lb-rulingset} with bridge edges $(u_2,u_{2}')$ and $(v_{n/4},v_2')$. The disconnected graph $D$ is given by replacing the bridge edges with the dashed edges.}
\label{fig:messagelb}
\end{figure}

\vspace{-0.2in}
\section{Conclusion}
\label{sec:conc}
We studied symmetry breaking problems in the \textsc{Congest} model and presented time- and message-efficient algorithms for ruling sets.
Several key open questions remain. First, can the MIS lower bounds in the \textsc{Local} model
shown by Kuhn et al.~\cite{KuhnMoscibrodaWattenhoferFull} be extended to 2-ruling sets?
In an orthogonal direction, can we derive time lower bounds for MIS in the \textsc{Congest} model, that are stronger than their \textsc{Local}-model counterparts?
And on the algorithms side, can we build on the techniques
presented here to improve the time bounds in the  \textsc{Congest} model? 
For example, can we solve the 2-ruling set problem in $O(\log^\alpha n)$ rounds for
some constant $\alpha < 1$.

Second, although we have presented near-tight message bounds for 2-ruling sets, we don't have a good understanding 
of the message-time {\em tradeoffs}.
In particular, a key question is whether we can design a 2-ruling set algorithm that uses 
$O(n \polylog n)$ messages, while running in $O(\polylog n)$ rounds? Such an algorithm will 
also lead to better algorithms for other distributed computing models
such as the $k$-machine model \cite{soda15}.
More generally, can we obtain tradeoff relationships that characterizes the dependence 
of one measure on the other or obtain lower bounds on the complexity of one measure while fixing the other measure.

\bibliography{mis}

\appendix
\section{Concentration Bounds} \label{sec:concentration}
\begin{theorem}[Chernoff Bound \cite{Upfalbook}] \label{thm:chernoff}
Let the random variables $X_1, X_2, \ldots, X_n$ be independently distributed in $[0, 1]$ and let $X = \sum_i X_i$.
Then, for $0 < \varepsilon < 1$,
$$Pr(X < (1-\varepsilon) E[X]) \le \exp\left(-\frac{\varepsilon^2}{2} E[X]\right).$$
\end{theorem}

\begin{theorem}[Bernstein's Inequality \cite{dubhashi}] \label{thm:bernstein}
Let the random variables $X_1, X_2, \ldots, X_n$ be independent with $X_i - E[X_i] \le b$ for each $i$, $1 \le i \le n$.
Let $X := \sum_i X_i$ and let $\sigma^2 := \sum_i \sigma_i^2$ be the variance of $X$.
Then for any $t > 0$,
$$Pr(X > E[X] + t) \le \exp\left(\frac{t^2}{2\sigma^2(1 + bt/3\sigma^2)}\right).$$
\end{theorem}

\end{document}